\documentclass[submission,copyright,creativecommons,hidelinks]{eptcs}
\providecommand{\event}{LSFA 2022} 
\usepackage{breakurl}             
\usepackage{underscore}           
\usepackage{cite}
\usepackage{amssymb}
\usepackage{amsmath}
\usepackage{amsthm}
\usepackage{bussproofs}
\usepackage{multicol}
\usepackage{comment}
\usepackage{enumitem}
\usepackage{url}
\usepackage{xcolor}
\usepackage{tabularx}
\usepackage{framed}
\usepackage{float}
\usepackage[figurename=Figure]{caption}
\usepackage{extarrows}
\usepackage{longtable}

\newtheorem{thm}{Theorem}

\newtheorem{lem}[thm]{Lemma}
\newtheorem{defn}[thm]{Definition}
\newtheorem{cor}[thm]{Corollary}
\newtheorem{exam}{Example}

\DeclareMathAlphabet{\mathrmbf}{\encodingdefault}{\rmdefault}{bx}{n}
\title{Equational Theorem Proving for Clauses over Strings}

\author{Dohan Kim
\institute{A. I. Research Lab, Kyungwon Plaza 201, Sujeong-gu, Seongnam-si, Gyeonggi-do, South Korea}
\email{dkim@airesearch.kr}
}

\def\titlerunning{Equational Theorem Proving for Clauses over Strings}
\def\authorrunning{D. Kim}

\begin{document}
\maketitle

\begin{abstract}
Although reasoning about equations over strings has been extensively studied for several decades, little research has been done for equational reasoning on general clauses over strings. This paper introduces a new superposition calculus with strings and present an equational theorem proving framework for clauses over strings. It provides a saturation procedure for clauses over strings and show that the proposed superposition calculus with contraction rules is refutationally complete. This paper also presents a new decision procedure for word problems over strings w.r.t.\;a set of conditional equations $R$ over strings if $R$ can be finitely saturated under the proposed inference system.
\end{abstract}

\section{Introduction}
Strings are fundamental objects in mathematics and many fields of science including computer science and biology. Reasoning about equations over strings has been widely studied in the context of string rewriting systems, formal language theory, word problems in semigroups, monoids and groups~\cite{Book1993,Epstein1992}, etc. Roughly speaking, reasoning about equations over strings replaces equals by equals w.r.t.\;a given reduction ordering $\succ$. For example, if we have two equations over strings $u_1u_2u_3\approx s$ and $u_2\approx t$ with $u_1u_2u_3\succ s$ and $u_2\succ t$, where $u_2$ is not the empty string, then we may infer the equation $u_1tu_3 \approx s$ by replacing $u_2$ in $u_1u_2u_3\approx s$ with $t$. Meanwhile, if we have two equations over strings $u_1u_2\approx s$ and $u_2u_3\approx t$ with $u_1u_2\succ s$ and $u_2u_3\succ t$, where $u_2$ is not the empty string, then we  should also be able to infer the equation $u_1t\approx su_3$. This can be done by  concatenating $u_3$ to both sides of $u_1u_2\approx s$ (i.e., $u_1u_2u_3\approx su_3$) and then replacing $u_2u_3$ in $u_1u_2u_3\approx su_3$ with $t$. Here, the \emph{monotonicity property} of equations over strings is assumed, i.e.,\;$s\approx t$ implies $usv\approx utv$ for strings $s$, $t$, $u$, and $v$.\footnote{Note that it suffices to assume the right monotonicity property of equations over strings, i.e.,\;$s\approx t$ implies $su\approx tu$ for strings $s$, $t$, and $u$, when finding overlaps between equations over strings under the monotonicity assumption.}\\
\indent This reasoning about equations over strings is the basic ingredient for \emph{completion}~\cite{Book1993,Holt2005} of string rewriting systems. A completion procedure~\cite{Book1993,Holt2005} attempts to construct a finite convergent string rewriting system, where a finite convergent string rewriting system provides a decision procedure for its corresponding equational theory.\\
\indent Unlike reasoning about equations over strings, equational reasoning on general clauses over strings has not been well studied, where clauses are often the essential building blocks for logical statements.\\
\indent This paper proposes a superposition calculus and an equational theorem proving procedure with clauses over strings. The results presented here generalize the results about completion of equations over strings~\cite{Holt2005,Book1993}. Throughout this paper, the monotonicity property of equations over strings is assumed and considered in the proposed inference rules. This assumption is natural and common to equations over strings occurring in algebraic structures (e.g., semigroups and monoids), formal language theory, etc. The \emph{cancellation property} of equations over strings is not assumed, i.e., $su \approx tu$ implies $s\approx t$ for strings $s$, $t$, and a nonempty string $u$ (cf. \emph{non}-\emph{cancellative}~\cite{Book1993} algebraic structures).\\
\indent Now, the proposed superposition inference rule is given roughly as follows:\\\\
\noindent \LeftLabel{\bf Superposition:\;\;}
\AxiomC{$C\vee u_1u_2 \approx s$}\AxiomC{$D\vee u_2u_3\approx t$}
\BinaryInfC{$C\vee D\vee u_1t\approx su_3$}
\DisplayProof\\\\
if $u_2$ is not the empty string, and $u_1u_2 \succ s$ and $u_2u_3 \succ t$.\\\\
\indent Intuitively speaking, using the monotonicity property, $C\vee u_1u_2u_3\approx su_3$ can be obtained from the left premise $C \vee u_1u_2 \approx s$. Then the above inference by Superposition can be viewed as an application of a conditional rewrite rule $D\vee u_2u_3\approx t$ to $C\vee u_1u_2u_3\approx su_3$, where $u_2u_3$ in $C \vee u_1u_2u_3\approx su_3$ is now replaced by $t$, and $D$ is appended to the conclusion. (Here, $D$ can be viewed as consisting of the positive and negative conditions.) Note that both $u_1$ and $u_3$ can be the empty string in the Superposition inference rule. These steps are combined into a single Superposition inference step. For example, suppose that we have three clauses 1: $ab\approx d$, 2: $bc\approx e$, and 3: $ae\not\approx dc$. We use the Superposition inference rule with 1 and 2, and obtain 4: $ae\approx dc$ from which we derive a contradiction with 3. The details of the inference rules in the proposed inference system are discussed in Section~\ref{sec:superposition}.\\
\indent The proposed superposition calculus is based on the simple string matching methods and the efficient length-lexicographic ordering instead of using equational unification and the more complex orderings, such as the lexicographic path ordering (LPO)~\cite{Dershowitz2001} and Knuth-Bendix ordering (KBO)~\cite{Baader1998}.\\
\indent This paper shows that a clause over strings can be translated into a clause over first-order terms, which allows one to use the existing notion of redundancy in the literature~\cite{Nieuwenhuis2001, Bachmair1994} for clauses over strings. Based on the notion of redundancy, one may delete redundant clauses using the contraction rules (i.e.,\;Simplification, Subsumption, and Tautology) during an equational theorem proving derivation in order to reduce the search space for a refutation.\\
\indent The \emph{model construction techniques}~\cite{Bachmair1994,Nieuwenhuis2001} is adapted for the refutational completeness of the proposed superposition calculus. This paper also uses a Herbrand interpretation by translating clauses over strings into clauses over first-order terms, where each nonground first-order clause represents all its ground instances. Note that this translation is not needed for the proposed inference system itself.\\
\indent Finally, the proposed equational theorem proving framework with clauses over strings allows one to provide a new decision procedure for word problems over strings w.r.t.\;a conditional equational theory $R$ if $R$ can be finitely saturated under the proposed inference system.

\section{Preliminaries}\label{sec:preliminaries}
It is assumed that the reader has some familiarity with equational theorem proving~\cite{Nieuwenhuis2001,Bachmair1994} and string rewriting systems~\cite{Book1993, Holt2005,Kapur1985}. The notion of conditional equations and Horn clauses are discussed in~\cite{Dershowitz1991}.\\
\indent An \emph{alphabet} $\Sigma$ is a finite set of symbols (or letters). The set of all strings of symbols over $\Sigma$ is denoted $\Sigma^*$ with the empty string $\lambda$.\\
\indent If $s\in \Sigma^*$, then the \emph{length} of $s$, denoted $|s|$, is defined as follows: $|\lambda|:=0$, $|a|:=1$ for each $a\in \Sigma$, and $|sa|:=|s|+1$ for $s\in \Sigma^*$ and $a\in\Sigma$.\\
\indent A \emph{multiset} is an unordered collection with possible duplicate elements. We denote by $M(x)$ the number of occurrences of an object $x$ in a multiset $M$.\\
\indent An \emph{equation} is an expression $s\approx t$, where $s$ and $t$ are strings, i.e., $s,t \in \Sigma^*$. A \emph{literal} is either a positive equation $L$, called a \emph{positive literal}, or a negative equation $\neg L$, called a \emph{negative literal}. We also write a negative literal $\neg (s\approx t)$ as $s\not\approx t$. We identify a positive literal $s\approx t$ with the multiset $\{\{s\},\{t\}\}$ and a negative literal $s\not\approx t$ with the multiset $\{\{s,t\}\}$. A \emph{clause} (over $\Sigma^*$) is a finite multiset of literals, written as a disjunction of literals $\neg A_1\vee \cdots \vee \neg A_m \vee B_1\vee \cdots \vee B_n$ or as an implication $\Gamma \rightarrow \Delta$, where $\Gamma=A_1 \wedge \cdots \wedge A_m$ and $\Delta=B_1\vee \cdots \vee B_n$. We say that $\Gamma$ is the \emph{antecedent} and $\Delta$ is the \emph{succedent} of clause $\Gamma \rightarrow \Delta$. A \emph{Horn clause} is a clause with at most one positive literal. The \emph{empty clause}, denoted $\square$, is the clause containing no literals.\\
\indent A \emph{conditional equation} is a clause of the form $(s_1 \approx t_1 \wedge \cdots \wedge s_n \approx t_n) \rightarrow l \approx r$. If $n=0$, a conditional equation is simply an equation. A conditional equation is naturally represented by a Horn clause. A \emph{conditional equational theory} is a set of conditional equations.\\
\indent Any ordering $\succ_S$ on a set $S$ can be extended to an ordering $\succ_S^{mul}$ on finite multisets over $S$ as follows: $M\succ_S^{mul} N$ if (i) $M\neq N$ and (ii) whenever $N(x) > M(x)$ then $M(y) > N(y)$, for some $y$ such that $y\succ_S x$.\\
\indent Given a multiset $M$ and an ordering $\succ$ on $M$, we say that $x$ is \emph{maximal} (resp.\;\emph {strictly maximal}) in $M$ if there is no $y\in M$ (resp.\;$y \in M\setminus\{x\})$ with $y\succ x$ (resp.\;$y \succ x$ or $x=y$).\\
\indent An ordering $>$ on $\Sigma^*$ is \emph{terminating} if there is no infinite chain of strings $s > s_1 > s_2 > \cdots$ for any $s \in \Sigma^*$. An ordering $>$ on $\Sigma^*$ is \emph{admissible} if $u > v$ implies $xuy > xvy$ for all $u,v,x,y \in \Sigma^*$.  An ordering $>$ on $\Sigma^*$  is  a \emph{reduction ordering} if it is terminating and admissible.\\
\indent The \emph{lexicographic ordering} $\succ_{lex}$ induced by a total precedence ordering $\succ_{prec}$ on $\Sigma$ ranks strings of the same length in $\Sigma^*$ by comparing the letters in the first index position where two strings differ using $\succ_{prec}$. For example, if $a=a_1a_2\cdots a_k$ and $b=b_1b_2\cdots b_k$, and the first index position where $a$ and $b$ are differ is $i$, then $a \succ_{lex} b$ if and only if $a_i \succ_{prec} b_i$.\\
\indent The \emph{length}-\emph{lexicographic ordering} $\succ$ on $\Sigma^*$ is defined as follows: $s \succ t$ if and only if $|s|>|t|$, or they have the same length and $s\succ_{lex} t$ for $s,t\in \Sigma^*$. If $\Sigma$ and $\succ_{prec}$ are fixed, then it is easy to see that we can determine whether $s\succ t$ for two (finite) input strings $s\in \Sigma^*$ and $t\in \Sigma^*$ in $O(n)$ time, where $n=|s|+|t|$. The length-lexicographic ordering $\succ$ on $\Sigma^*$ is a reduction ordering. We also write $\succ$ for a multiset extension of $\succ$ if it is clear from context.\\
\indent We say that $\approx$ has the \emph{monotonicity property} over $\Sigma^*$ if $s\approx t$ implies $usv \approx utv$ for all $s,t,u,v\in\Sigma^*$. Throughout this paper, it is assumed that $\approx$ has the monotonicity property over $\Sigma^*$.

\section{Superposition with Strings}\label{sec:superposition}
\subsection{Inference Rules}
The following inference rules for clauses over strings are parameterized by a selection function $\mathcal{S}$ and the length-lexicographic ordering $\succ$, where $\mathcal{S}$ arbitrarily selects exactly one negative literal for each clause containing at least one negative literal (see Section 3.6 in~\cite{Nieuwenhuis2001} or Section 6 in~\cite{Bachmair1998}). In this strategy, an inference involving a clause with a selected literal is performed before an inference from clauses without a selected literal for a theorem proving process. The intuition behind the (eager) selection of negative literals is that, roughly speaking, one may first prove the whole antecedent of each clause from other clauses. Then clauses with no selected literals are involved in the main deduction process. This strategy is particularly useful when we consider Horn completion in Section~\ref{sec:HornCompletion} and a decision procedure for the word problems associated with it. In the following, the symbol $\bowtie$ is used to denote either $\approx$ or $\not\approx$.\\

\noindent \LeftLabel{\bf Superposition:\;\;}
\AxiomC{$C\vee u_1u_2 \approx s$}\AxiomC{$D\vee u_2u_3\approx t$}
\BinaryInfC{$C\vee D\vee u_1t\approx su_3$}
\DisplayProof\\\\
if (i) $u_2$ is not $\lambda$, (ii) $C$ contains no selected literal, (iii) $D$ contains no selected literal, (iv) $u_1u_2 \succ s$, and (v) $u_2u_3 \succ t$.\footnote{We do not require that $u_1u_2\approx s$ (resp.\;$u_2u_3\approx t$) is strictly maximal in the left premise (resp. the right premise) because of the assumption on the monotonicity property of equations over strings (see also Lemma~\ref{lem:representation} in Section~\ref{subsec:lifting}).}

\noindent \LeftLabel{\bf Rewrite:\;\;}
\AxiomC{$C\vee u_1u_2u_3 \bowtie s$}\AxiomC {$D\vee u_2\approx t$}
\BinaryInfC{$C\vee D\vee u_1tu_3\bowtie s$}
\DisplayProof\\\\
if (i) $u_1u_2u_3\bowtie s$ is selected for the left premise whenever $\bowtie$ is $\not\approx$, (ii) $C$ contains no selected literal whenever $\bowtie$ is $\approx$, (iii) $D$ contains no selected literal, and (iv) $u_2 \succ t$.\footnote{Note that $u_2\succ t$ implies that $u_2$ cannot be the empty string $\lambda$.}\\

\noindent \LeftLabel{\bf Equality Resolution:\;\;}
\AxiomC{$C\vee s\not\approx s$}
\UnaryInfC{$C$}
\DisplayProof\\\\
if $s\not\approx s$ is selected for the premise.\\

The following Paramodulation and Factoring inference rules are used for non-Horn clauses containing positive literals only (cf.\;\emph{Equality Factoring}~\cite{Bachmair1994,Nieuwenhuis2001} and \emph{Merging Paramodulation} rule~\cite{Bachmair1994}).\\

\noindent \LeftLabel{\bf Paramodulation:\;\;}
\AxiomC{$C\vee s\approx u_1u_2$}\AxiomC {$D\vee u_2u_3\approx t$}
\BinaryInfC{$C\vee D \vee su_3\approx u_1t$}
\DisplayProof\\\\
if (i) $u_2$ is not $\lambda$, (ii) $C$ contains no selected literal, (iii) $C$ contains a positive literal, (iv) $D$ contains no selected literal, (v) $s\succ u_1u_2$, and (vi) $u_2u_3\succ t$.\\


\noindent \LeftLabel{\bf Factoring:\;\;}
\AxiomC{$C\vee s\approx t \vee su\approx tu$}
\UnaryInfC{$C\vee su\approx tu$}
\DisplayProof\\\\
if $C$ contains no selected literal.\\

In the proposed inference system, finding whether a string $s$ occurs within a string $t$ can be done in linear time in the size of $s$ and $t$ by using the existing string matching algorithms such as the Knuth-Morris-Pratt (KMP) algorithm~\cite{Cormen2001}. For example, the KMP algorithm can be used for finding $u_2$ in $u_1u_2u_3$ in the Rewrite rule and finding $u_2$ in $u_1u_2$ in the Superposition and Paramodulation rule.

In the remainder of this paper, we denote by $\mathfrak{S}$ the inference system consisting of the Superposition, Rewrite, Equality Resolution, Paramodulation and the Factoring rule, and denote by $S$ a set of clauses over strings. Also, by the \emph{contraction rules} we mean the following inference rules--Simplification, Subsumption and Tautology.\\

\noindent \LeftLabel{\bf Simplification:\;\;}
\AxiomC{$S\cup\{C\vee l_1ll_2\bowtie v,\; l\approx r\}$}
\UnaryInfC{$S\cup\{C\vee l_1rl_2\bowtie v,\; l\approx r\}$}
\DisplayProof\\\\
if (i) $l_1ll_2\bowtie v$ is selected for $C\vee l_1ll_2\bowtie v$ whenever $\bowtie$ is $\not\approx$, (ii) $l_1$ is not $\lambda$, and (iii) $l \succ r$.\\

In the following inference rule, we say that a clause $C$ \emph{subsumes} a clause $C^\prime$ if $C$ is contained in $C'$, where $C$ and $C'$ are viewed as the finite multisets.\\

\noindent \LeftLabel{\bf Subsumption:\;\;}
\AxiomC{$S\cup\{C, C^\prime\}$}
\UnaryInfC{$S\cup\{C\}$}
\DisplayProof\\\\
if $C\subseteq C'$.

\noindent \LeftLabel{\bf Tautology:\;\;}
\AxiomC{$S\cup\{C \vee s\approx s\}$}
\UnaryInfC{$S$}
\DisplayProof

\begin{exam}\normalfont\label{ex:ex1}
Let $a \succ b \succ c\succ d \succ e$ and consider the following inconsistent set of clauses 1: $ad\approx b\vee ad\approx c$, 2: $b \approx c$, 3: $ad\approx e$, and 4: $c\not\approx e$. Now, we show how the empty clause is derived:\\
5: $ad \approx c \vee ad \approx c$ (Paramodulation of 1 with 2)\\
6: $ad \approx c$ (Factoring of 5)\\
7: $c \approx e$ (Rewrite of 6 with 3)\\
8: $e \not\approx e$ ($c\not\approx e$ is selected for 4. Rewrite of 4 with 7)\\
9: $\square$ ($e \not\approx e$ is selected for 8. Equality Resolution on 8)\\\\
Note that there is no inference with the selected literal in 4 from the initial set of clauses 1, 2, 3, and 4. We produced clauses 5, 6, and 7 without using a selected literal. Once we have clause 7, there is an inference with the selected literal in 4.
\end{exam}

\begin{exam}\normalfont
\label{ex:ex2}
Let $a \succ b \succ c \succ d$ and consider the following inconsistent set of clauses 1: $aa \approx a \vee bd\not\approx a$, 2: $cd \approx b$, 3: $ad \approx c$, 4: $bd\approx a$, and 5: $dab\not\approx db$. Now, we show how the empty clause is derived:\\
6: $aa \approx a \vee a\not\approx a$ ($bd\not\approx a$ is selected for 1. Rewrite of 1 with 4)\\
7: $aa \approx a$ ($a\not\approx a$ is selected for 6. Equality resolution on 6)\\
8: $ac\approx ad$ (Superposition of 7 with 3)\\
9: $add \approx ab$ (Superposition of 8 with 2)\\
10: $ab \approx cd$ (Rewrite of 9 with 3)\\
11: $dcd\not\approx db$ ($dab\not\approx db$ is selected for 5. Rewrite of 5 with 10)\\
12: $db \not\approx db$ ($dcd\not\approx db$ is selected for 11. Rewrite of 11 with 2)\\
13: $\square$ ($db \not\approx db$ is selected for 12. Equality Resolution on 12)
\end{exam}

\subsection{Lifting Properties}\label{subsec:lifting}
Recall that $\Sigma^*$ is the set of all strings over $\Sigma$ with the empty string $\lambda$. We let $T(\Sigma\cup \{\bot\})$ be the set of all first-order ground terms over $\Sigma\cup \{\bot\}$, where each letter from $\Sigma$ is interpreted as a unary function symbol and $\bot$ is the only constant symbol. (The constant symbol $\bot$ does not have a special meaning (e.g.,\;``false'') in this paper.) We remove parentheses for notational convenience for each term in $T(\Sigma\cup \{\bot\})$. Since $\bot$ is the only constant symbol, we see that $\bot$ occurs only once at the end of each term in $T(\Sigma\cup \{\bot\})$. We may view each term in $T(\Sigma\cup \{\bot\})$ as a string ending with $\bot$. Now, the definitions used in Section~\ref{sec:preliminaries} can be carried over to the case when $\Sigma^*$ is replaced by $T(\Sigma\cup \{\bot\})$. In the remainder of this paper, we use the string notation for terms in $T(\Sigma\cup \{\bot\})$ unless otherwise stated.\\
\indent  Let $s\approx t$ be an equation over $\Sigma^*$. Then we can associate $s\approx t$ with the equation $s(x)\approx t(x)$, where $s(x)\approx t(x)$ represents the set of all its ground instances over $T(\Sigma\cup \{\bot\})$. (Here, $\lambda(x)$ and $\lambda\bot$ correspond to $x$ and $\bot$, respectively.) First, $s\approx t$ over $\Sigma^*$ corresponds to $s\bot \approx t\bot$ over $T(\Sigma\cup \{\bot\})$. Now, using the monotonicity property, if we concatenate string $u$ to both sides of $s\approx t$ over $\Sigma^*$, then we have $su \approx tu$, which corresponds to $su\bot\approx tu\bot$.\\
\indent There is a similar approach in string rewriting systems. If $S$ is a string rewriting system over $\Sigma^*$, then it is known that we can associate term rewriting system $R_S$ with $S$ in such a way that $R_S:=\{l(x)\rightarrow r(x)\,|\,l\rightarrow r \in S\}$~\cite{Book1993}, where $x$ is a variable and each letter from $\Sigma$ is interpreted as a unary function symbol. We may rename variables (by standardizing variables apart) whenever necessary. This approach is particularly useful when we consider critical pairs between the rules in a string rewriting system. For example, if there are two rules $aa\rightarrow c$ and $ab\rightarrow d$ in $S$, then we have $cb \leftarrow aab\rightarrow ad$, where ${<}cb, ad{>}$ (or  ${<}ad, cb{>}$) is a \emph{critical pair} formed from these two rules. This critical pair can also be found if we associate $aa\rightarrow c\in S$ with $a(a(x))\rightarrow c(x) \in R_S$ and $ab\rightarrow d\in S$ with $a(b(x))\rightarrow d(x) \in R_S$. First, we rename the rule $a(b(x))\rightarrow d(x) \in R_S$ into $a(b(y))\rightarrow d(y)$. Then by mapping $x$ to $b(z)$ and $y$ to $z$, we have $c(b(z))\leftarrow a(a(b(z)))\rightarrow a(d(z))$, where ${<}c(b(z)), a(d(z)){>}$ is a critical pair formed from these two rules. This critical pair can be associated with the critical pair ${<}cb, ad{>}$ formed from $aa\rightarrow c$ in $S$ and $ab\rightarrow d$ in $S$.\\ 
\indent However, if $s\not\approx t$ is a negative literal over strings, then we cannot simply associate $s\not\approx t$ with the negative literal $s(x)\not\approx t(x)$ over first-order terms. Suppose to the contrary that we associate $s\not\approx t$ with $s(x)\not\approx t(x)$. Then $s\not\approx t$ implies $su\not\approx tu$ for a nonempty string $u$ because we can substitute $u(y)$ for $x$ in $s(x)\not\approx t(x)$, and $su\not\approx tu$ can also be associated with $s(u(y))\not\approx t(u(y))$. Using the contrapositive argument, this means that $su\approx tu$ implies $s\approx t$ for the nonempty string $u$. Recall that we do not assume the cancellation property of equations over strings in this paper.\footnote{One may assume the cancellation property and associate $s\not\approx t$ over strings with $s(x)\not\approx t(x)$ over first-order terms, which is beyond the scope of this paper.} Instead, we simply associate $s\not\approx t$ with $s\bot\not\approx t\bot$. The following lemma is based on the above observations. We denote by $T(\Sigma\cup \{\bot\}, X)$ the set of first-order terms built on $\Sigma\cup\{\bot\}$ and a denumerable set of variables $X$, where each symbol from $\Sigma$ is interpreted as a unary function symbol and $\bot$ is the only constant symbol.

\begin{lem} \label{lem:representation} 
Let $C:=s_1\approx t_1\vee\cdots\vee s_m \approx t_m \vee u_1\not\approx v_1 \vee\cdots \vee u_n \not\approx v_n$ be a clause over $\Sigma^*$ and $P$ be the set of all clauses that follow from $C$ using the monotonicity property. Let $Q$ be the set of all ground instances of the clause $s_1(x_1)\approx t_1(x_1) \vee \cdots \vee s_m(x_m)\approx t_m(x_m) \vee u_1\bot\not\approx v_1\bot \vee \cdots \vee u_n\bot\not\approx v_n\bot$ over $T(\Sigma\cup \{\bot\}, X)$, where $x_1,\ldots, x_m$ are distinct variables in $X$ and each letter from $\Sigma$ is interpreted as a unary function symbol. Then there is a one-to-one correspondence between $P$ and $Q$. 
\end{lem}
\begin{proof}
For each element $D$ of $P$, $D$ has the form $D:=s_1w_1\approx t_1w_1\vee\cdots\vee s_mw_m \approx t_mw_m \vee u_1\not\approx v_1 \vee\cdots \vee u_n \not\approx v_n$ for some $w_1,\ldots, w_m \in \Sigma^*$. (If $w_i=\lambda$ for all $1\leq i \leq m$, then $D$ is simply $C$.) Now, we map each element $D$ of $P$ to $D'$ in $Q$, where $D':=s_1w_1\bot\approx t_1w_1\bot \vee\cdots\vee s_mw_m\bot\approx t_mw_m\bot \vee u_1\bot\not\approx v_1\bot \vee\cdots \vee u_n\bot \not\approx v_n\bot$. Since $\bot$ is the only constant symbol in $\Sigma\cup \{\bot\}$, it is easy to see that this mapping is well-defined and bijective.
\end{proof}

\begin{defn}\label{defn:defn1}\normalfont
(i) We say that every term in $T(\Sigma\cup \{\bot\})$ is a $g$-\emph{term}. (Recall that we remove parentheses for notational convenience.)\\
(ii) Let $s\approx t$ (resp.\;$s\rightarrow t$) be an equation (resp. a rule) over $\Sigma^*$. We say that $su\bot\approx tu\bot$ (resp.\;$su\bot\rightarrow tu\bot$) for some string $u$ is a $g$-\emph{equation} (resp.\;a $g$-\emph{rule}) of $s\approx t$ (resp.\;$s\rightarrow t$).\\
(iii) Let $s\not\approx t$ be a negative literal over $\Sigma^*$. We say that $s\bot\not\approx t\bot$ is a (negative) $g$-\emph{literal} of $s\not\approx t$.\\
(iv) Let $C:=s_1\approx t_1\vee\cdots\vee s_m\approx t_m \vee u_1\not\approx v_1 \vee \cdots\vee u_n\not\approx v_n$ be a clause over $\Sigma^*$. We say that $s_1w_1\bot\approx t_1w_1\bot\vee\cdots\vee s_mw_m\bot\approx t_mw_m\bot \vee u_1\bot\not\approx v_1\bot \vee \cdots\vee u_n\bot\not\approx v_n\bot$ for some strings $w_1,\ldots,w_m$ is a $g$-\emph{clause} of clause $C$. Here, each $w_k\bot \in T(\Sigma\cup \{\bot\})$ for nonempty string $w_k$ in the $g$-clause is said to be a \emph{substitution part} of $C$.\\
(v) Let $\pi$ be an inference (w.r.t.\;$\mathfrak{S}$) with premises $C_1,\ldots, C_k$ and conclusion $D$. Then a $g$-\emph{instance} of $\pi$ is an inference (w.r.t.\;$\mathfrak{S}$) with premises $C_1',\ldots, C_k'$ and conclusion $D'$, where $C_1',\ldots, C_k'$ and $D'$ are $g$-clauses of $C_1,\ldots, C_k$ and $D$, respectively.
\end{defn}

Since each term in $T(\Sigma\cup \{\bot\})$ is viewed as a string, we may consider inferences between $g$-clauses using $\mathfrak{S}$. Note that concatenating a (nonempty) string at the end of a $g$-term is not allowed for any $g$-term over $T(\Sigma\cup \{\bot\})$. For example, $abc\bot d$ is not a $g$-term, and $a\bot \not\approx b\bot \vee abc\bot d \approx def\bot d$ is not a $g$-clause. We emphasize that we are only concerned with inferences between (legitimate) $g$-clauses here.\\
\indent We may also use the length-lexicographic ordering $\succ_g$ on $g$-terms. Given a total precedence ordering on $\Sigma\cup \{\bot\}$ for which $\bot$ is minimal, it can be easily verified that $\succ_g$ is a total reduction ordering on $T(\Sigma\cup \{\bot\})$. We simply denote the multiset extension $\succ_g^{mul}$ of $\succ_g$ as $\succ_g$ for notational convenience. Similarly, we denote ambiguously all orderings on $g$-terms, $g$-equations, and $g$-clauses over $T(\Sigma\cup \{\bot\})$ by $\succ_g$. Now, we consider the lifting of inferences of $\mathfrak{S}$ between $g$-clauses over $T(\Sigma\cup \{\bot\})$ to inferences of $\mathfrak{S}$ between clauses over $\Sigma^*$. Let $C_1,\ldots, C_n$ be clauses over $\Sigma^*$ and let
\begin{prooftree}
\AxiomC{$C_1'\ldots C_n'$}
\UnaryInfC{$C'$}
\end{prooftree}
be an inference between their $g$-clauses, where $C_i'$ is a $g$-clause of $C_i$ for all $1\leq i \leq n$. We say that this inference between $g$-clauses can be \emph{lifted} if there is an inference
\begin{prooftree}
\AxiomC{$C_1\ldots C_n$}
\UnaryInfC{$C$}
\end{prooftree}
such that $C'$ is a $g$-clause of $C$. In what follows, we assume that a $g$-literal $L'_i$ in $C'_i$ is selected in the same way as $L_i$ in $C_i$, where $L_i$ is a negative literal in $C_i$ and $L'_i$ is a $g$-literal of $L_i$.

Lifting of an inference between $g$-clauses is possible if it does not correspond to a $g$-instance of an inference (w.r.t.\;$\mathfrak{S}$) into a substitution part of a clause, which is not necessary (see~\cite{Bachmair1995, Nieuwenhuis2001}). Suppose that there is an inference between $g$-clauses $C_1'\ldots C_n'$ with conclusion $C'$ and there is also an inference between clauses $C_1\ldots C_n$ over $\Sigma^*$ with conclusion $C$, where $C_i'$ is a $g$-clause of $C_i$ for all $1\leq i \leq n$. Then, the inference between $g$-clauses $C_1'\ldots C_n'$ over $T(\Sigma\cup \{\bot\})$ can be lifted to the inference between clauses $C_1\ldots C_n$ over $\Sigma^*$ in such a way that $C'$ is a $g$-clause of $C$. This can be easily verified for each inference rule in $\mathfrak{S}$.

\begin{exam}\normalfont
Consider the following Superposition inference with $g$-clauses:
\begin{prooftree}
\AxiomC{$ad\bot \approx cd\bot \vee aabb\bot\approx cbb\bot $}\AxiomC{$abb\bot \approx db\bot$}
\BinaryInfC{$ad\bot \approx cd\bot \vee adb\bot \approx cbb\bot$}
\end{prooftree}
where $ad\bot \approx cd\bot \vee aabb\bot\approx cbb\bot$ (resp.\;$abb\bot \approx db\bot$) is a $g$-clause of $a\approx c \vee aa\approx c$ (resp.\;$ab \approx d$) and $aabb\bot \succ_g cbb\bot$ (resp.\;$abb\bot \succ_g db\bot$). This Superposition inference between $g$-clauses can be lifted to the following Superposition inference between clauses over $\Sigma^*$:
\begin{prooftree}
\AxiomC{$a\approx c \vee aa\approx c$}\AxiomC{$ab \approx d$}
\BinaryInfC{$a\approx c \vee ad \approx cb$}
\end{prooftree}
where $aa\succ c$ and $ab \succ d$. We see that conclusion $ad\bot \approx cd\bot \vee adb\bot \approx cbb\bot$ of the Superposition inference between the above $g$-clauses is a $g$-clause of conclusion $a\approx c \vee ad \approx cb$ of this inference.
\end{exam}
\begin{exam}\normalfont
Consider the following Rewrite inference with $g$-clauses:
\begin{prooftree}
\AxiomC{$a\bot \not\approx d\bot \vee aabb\bot\not\approx cd\bot$}\AxiomC{$abb\bot \approx cb\bot$}
\BinaryInfC{$a\bot \not\approx d\bot \vee acb\bot \not\approx cd\bot$}
\end{prooftree}
where $aabb\bot\not\approx cd\bot$ is selected and $a\bot \not\approx d\bot \vee aabb\bot\not\approx cd\bot$ (resp.\;$abb\bot \approx cb\bot$) is a $g$-clause of $a\not\approx d \vee aabb\not\approx cd$ (resp.\;$ab \approx c$) with $abb\bot \succ_g cb\bot$. This Rewrite inference between $g$-clauses can be lifted to the following Rewrite inference between clauses over $\Sigma^*$:
\begin{prooftree}
\AxiomC{$a \not\approx d \vee aabb\not\approx cd$}\AxiomC{$ab \approx c$}
\BinaryInfC{$a \not\approx d  \vee acb \not\approx cd$}
\end{prooftree}
where $aabb\not\approx cd$ is selected and $ab\succ c$. We see that conclusion $a\bot \not\approx d\bot \vee acb\bot \not\approx cd\bot$ of the Rewrite inference between the above $g$-clauses is a $g$-clause of conclusion $a \not\approx d  \vee acb \not\approx cd$ of this inference.
\end{exam}

\section{Redundancy and Contraction Techniques}
By Lemma~\ref{lem:representation} and Definition~\ref{defn:defn1}, we may translate a clause $C:=s_1\approx t_1\vee\cdots\vee s_m \approx t_m \vee u_1\not\approx v_1 \vee\cdots \vee u_n \not\approx v_n$ over $\Sigma^*$ with all its implied clauses using the monotonicity property into the clause $s_1(x_1)\approx t_1(x_1) \vee \cdots \vee s_m(x_m)\approx t_m(x_m) \vee u_1\bot\not\approx v_1\bot \vee \cdots \vee u_n\bot\not\approx v_n\bot$ over $T(\Sigma\cup \{\bot\}, X)$ with all its ground instances, where $x_1,\ldots, x_m$ are distinct variables in $X$, each symbol from $\Sigma$ is interpreted as a unary function symbol, and $\bot$ is the only constant symbol. This allows us to adapt the existing notion of redundancy found in the literature~\cite{Nieuwenhuis2001, Bachmair1994}.

\begin{defn}\label{defn:Herbrand}\normalfont 
(i) Let $R$ be a set of $g$-equations or $g$-rules. Then the congruence $\leftrightarrow_R^*$ defines an equality \emph{Herbrand Interpretation} $I$, where the domain of $I$ is $T(\Sigma\cup \{\bot\})$. Each unary function symbol $s\in \Sigma$ is interpreted as the unary function $s_I$, where $s_I(u\bot)$ is the $g$-term $su\bot$. (The constant symbol $\bot$ is simply interpreted as the constant $\bot$.) The only predicate $\approx$ is interpreted by $s\bot\approx t\bot$ if $s\bot\leftrightarrow_R^*t\bot$. We denote by $R^*$ the interpretation $I$ defined by $R$ in this way. $I$ \emph{satisfies} (is a \emph{model} of) a $g$-clause $\Gamma \rightarrow \Delta$, denoted by $I\models \Gamma \rightarrow \Delta$, if $I \not\supseteq \Gamma$ or $I \cap \Delta \neq \emptyset$. In this case, we say that $\Gamma \rightarrow \Delta$ is \emph{true} in $I$. We say that $I$ \emph{satisfies} a clause $C$ over $\Sigma^*$ if $I$ satisfies all $g$-clauses of $C$. We say that $I$ \emph{satisfies} a set of clauses $S$ over $\Sigma^*$, denoted by $I \models S$, if $I$ satisfies every clause in $S$.\\
(ii) A $g$-clause $C$  \emph{follows} from a set of $g$-clauses $\{C_1,\ldots,C_n\}$, denoted by $\{C_1,\ldots,C_n\}\models C$, if $C$ is true in every model of $\{C_1, \ldots, C_k\}$.
\end{defn}

\begin{defn}\label{defn:redundancy}\normalfont Let $S$ be a set of clauses over $\Sigma^*$.\\ 
(i) A $g$-clause $C$ is \emph{redundant} w.r.t.\;$S$ if there exist $g$-clauses $C_1', \ldots, C_k'$ of clauses $C_1, \ldots, C_k$ in $S$, such that $\{C_1',\ldots,C_k'\}\models C$ and $C\succ_g C_i'$ for all $1\leq i \leq k$. A clause in $S$ is \emph{redundant} w.r.t.\;$S$ if all its $g$-clauses are redundant  w.r.t.\;$S$.\\
(ii) An inference $\pi$ with conclusion $D$ is \emph{redundant} w.r.t.\;$S$ if for every $g$-instance of $\pi$ with maximal premise $C'$ (w.r.t.\;$\succ_g$) and conclusion $D'$, there exist $g$-clauses $C_1', \ldots, C_k'$ of clauses $C_1,\ldots, C_k$ in $S$ such that $\{C_1',\ldots,C_k'\}\models D'$ and $C' \succ_g C_i'$ for all $1\leq i \leq k$, where $D'$ is a $g$-clause of $D$.
\end{defn}

\begin{lem} If an equation $l\approx r$ simplifies a clause $C\vee l_1ll_2\bowtie v$ into $C\vee l_1rl_2\bowtie v$ using the Simplification rule, then $C\vee l_1ll_2\bowtie v$ is redundant w.r.t.\;$\{C\vee l_1rl_2\bowtie v, l\approx r\}$.
\end{lem}
\begin{proof}
\indent Suppose that $l\approx r$ simplifies $D:=C\vee l_1ll_2\not\approx v$ into $C\vee l_1rl_2\not\approx v$, where $l_1ll_2\not\approx v$ is selected for $D$. Then, every $g$-clause $D'$ of $D$ has the form $D':=C'\vee l_1ll_2\bot \not\approx v\bot$, where $C'$ is a $g$-clause of $C$. Now, we may infer that $\{D'', ll_2\bot \approx rl_2\bot\}\models D'$, where $D'':=C'\vee l_1rl_2\bot \not\approx v\bot$ is a $g$-clause of $C\vee l_1rl_2\not\approx v$ and $ll_2\bot \approx rl_2\bot$ is a $g$-equation of $l\approx r$. We also have $D' \succ_g D''$ and $D' \succ_g ll_2\bot \approx rl_2\bot$, and thus the conclusion follows.\\
\indent Otherwise, suppose that $l\approx r$ simplifies $D:=C\vee  l_1ll_2\approx v$ into $C\vee  l_1rl_2\approx v$. Then every $g$-clause $D'$ of $D$ has the form $D':=C'\vee l_1ll_2w\bot \approx vw\bot$ for some $w\in \Sigma^*$, where $C'$ is a $g$-clause of $C$. Now, we have $\{D'', ll_2w\bot \approx rl_2w\bot\}\models D'$, where $D'':=C'\vee l_1rl_2w\bot \approx vw\bot$ is a $g$-clause of $C\vee l_1rl_2\approx v$ for some $w\in\Sigma^*$ and $ll_2w\bot \approx rl_2w\bot$ is a $g$-equation of $l\approx r$. We also have $D' \succ_g D''$ and $D' \succ_g ll_2w\bot \approx rl_2w\bot$ because $l_1$ is not $\lambda$ in the condition of the rule (i.e., $l_1ll_2w\bot \succ_g ll_2w\bot$), and thus the conclusion follows.
\end{proof}

We see that if $C$ subsumes $C^\prime$ with $C$ and $C^\prime$ containing the same number of literals, then they are the same when viewed as the finite multisets, so we can remove $C^\prime$. Therefore, we exclude this case in the following lemma.

\begin{lem}
If a clause $C$ subsumes a clause $D$ and $C$ contains fewer literals than $D$, then $D$ is redundant w.r.t.\;$\{C\}$.
\end{lem}
\begin{proof}
Suppose that $C$ subsumes $D$ and $C$ contains fewer literals than $D$. Then $D$ can be denoted by $C\vee B$ for some nonempty clause $B$. Now, for every $g$-clause $D^{\prime}:=C'\vee B'$ of $D$, we have $\{C'\}\models D^{\prime}$ with $D^{\prime} \succ_g C'$, where $C'$ and $B'$ are $g$-clauses of $C$ and $B$, respectively. Thus, $D$ is redundant w.r.t.\;$\{C\}$.
\end{proof}

\begin{lem}
A tautology $C \vee s\approx s$ is redundant.
\end{lem}
\begin{proof}
It is easy to see that for every $g$-clause $C' \vee su\bot\approx su\bot$ of $C \vee s\approx s$, we have $\models C' \vee su\bot\approx su\bot$, where $u\in \Sigma^*$ and $C'$ is a $g$-clause of $C$. Thus, $C \vee s\approx s$ is redundant.
\end{proof}

\section{Refutational Completeness}\label{sec:completeness}
In this section, we adapt the model construction and equational theorem proving techniques used in~\cite{Bachmair1994, Nieuwenhuis2001, Kim2021} and show that $\mathfrak{S}$ with the contraction rules is refutationally complete.

\begin{defn}\normalfont A $g$-equation $s\bot \approx t\bot$ is \emph{reductive} for a $g$-clause $C:=D\vee s\bot\approx t\bot $ if $s\bot \approx t\bot$ is strictly maximal (w.r.t.\;$\succ_g$) in $C$ with $s\bot \succ_g t\bot$.
\end{defn}
 
\begin{defn}\normalfont\label{defn:model}(Model Construction)\;Let $S$ be a set of clauses over $\Sigma^*$. We use induction on $\succ_g$ to define the sets $R_C,E_C$, and $I_C$ for all $g$-clauses $C$ of clauses in $S$. Let $C$ be such a $g$-clause of a clause in $S$ and suppose that $E_{C^\prime}$ has been defined for all $g$-clauses $C^\prime$ of clauses in $S$ for which $C\succ_g C^\prime$. Then we define by $R_C=\bigcup_{C\succ_g C^\prime} E_{C^\prime}$. We also define by $I_C$ the equality interpretation $R_C^*$, which denotes the least congruence containing $R_C$.\\
\indent Now, let $C:=D \vee s\bot \approx t\bot$ such that $C$ is not a $g$-clause of a clause with a selected literal in $S$. Then $C$ produces $E_C=\{s\bot\rightarrow t\bot\}$ if the following conditions are met: (1) $I_C\not\models C$, (2) $I_C\not\models t\bot\approx t'\bot$ for every $s\bot\approx t'\bot$ in $D$, (3) $s\bot\approx t\bot$ is reductive for $C$, and (4) $s\bot$ is irreducible by $R_C$. We say that $C$ is \emph{productive} and produces $E_C$ if it satisfies all of the above conditions. Otherwise, $E_C=\emptyset$. Finally, we define $I_S$ as the equality interpretation $R_S^*$, where $R_S=\bigcup_C E_C$ is the set of all $g$-rules produced by $g$-clauses of clauses in~$S$.
\end{defn}

\begin{lem}\label{lem:confluent} (i) $R_S$ has the Church-Rosser property.\\
(ii) $R_S$ is terminating.\\
(iii) For $g$-terms $u\bot$ and $v\bot$, $I_S \models u\bot\approx v\bot$ if and only if $u\bot \downarrow_{R_S} v\bot$.\\
(iv) If $I_S\models s\approx t$, then $I_S \models usv \approx utv$ for nonempty strings $s,t,u,v \in \Sigma^*$.
\end{lem}
\begin{proof}
(i) $R_S$ is left-reduced because there are no overlaps among the left-hand sides of rewrite rules in $R_S$, and thus $R_S$ has the Church-Rosser property.\\
(ii) For each rewrite rule $l\bot \rightarrow r\bot$ in $R_S$, we have $l\bot \succ_g r\bot$, and thus $R_S$ is terminating.\\
(iii) Since $R_S$ has the Church-Rosser property and is terminating by (i) and (ii), respectively, $R_S$ is convergent. Thus, $I_S \models u\bot\approx v\bot$ if and only if $u\bot\downarrow_{R_S} v\bot$ for $g$-terms $u\bot$ and $v\bot$.\\
(iv) Suppose that $I_S\models s\approx t$ for nonempty strings $s$ and $t$. Then, we have $I_S\models svw\bot\approx tvw\bot$ for all strings $v$ and $w$ by Definition~\ref{defn:Herbrand}(i). Similarly, since $I_S$ is an equality Herbrand interpretation, we also have $I_S\models usvw\bot\approx utvw\bot$ for all strings $u$, which means that $I_S\models usv\approx utv$ by Definition~\ref{defn:Herbrand}(i).
\end{proof}

Lemma~\ref{lem:confluent}(iv) says that the monotonicity assumption used in this paper holds w.r.t.\;a model constructed by Definition~\ref{defn:model}.

\begin{defn}\normalfont\label{defn:Saturation}Let $S$ be a set of clauses over $\Sigma^*$. We say that $S$ is \emph{saturated} under $\mathfrak{S}$ if every inference by $\mathfrak{S}$ with premises in $S$ is redundant w.r.t.\;$S$.
\end{defn}

\begin{defn}\normalfont
Let $C:=s_1\approx t_1\vee\cdots\vee s_m\approx t_m \vee u_1\not\approx v_1 \vee \cdots\vee u_n\not\approx v_n$ be a clause over $\Sigma^*$, and $C'=s_1w_1\bot\approx t_1w_1\bot\vee\cdots\vee s_mw_m\bot\approx t_mw_m\bot \vee u_1\bot\not\approx v_1\bot \vee \cdots\vee u_n\bot\not\approx v_n\bot$ for some strings $w_1,\ldots,w_m$ be a $g$-clause of $C$. We say that $C'$ is a \emph{reduced} $g$-\emph{clause} of $C$ w.r.t.\;a rewrite system $R$ if every $w_i\bot$, $1\leq i \leq m$, is not reducible by $R$.
\end{defn}

In the proof of the following  lemma, we write $s[t]_{suf}$ to indicate that $t$ occurs in $s$ as a suffix and (ambiguously) denote by $s[u]_{suf}$ the result of replacing the occurrence of $t$ (as a suffix of $s$) by $u$.

\begin{lem}\label{lem:model} Let $S$ be saturated under $\mathfrak{S}$ not containing the empty clause and $C$ be a $g$-clause of a clause in $S$. Then $C$ is true in $I_S$. More specifically,\\
(i) If $C$ is redundant w.r.t.\;$S$, then it is true in $I_S$.\\
(ii)\,If $C$ is not a reduced $g$-clause of a clause in $S$ w.r.t.\,$R_S$,\,then it is true in $I_S$.\\
(iii) If $C:=C' \vee s\bot \approx t\bot$ produces the rule $s\bot \rightarrow t\bot$, then $C'$ is false and $C$ is true in $I_S$.\\
(iv) If $C$ is a $g$-clause of a clause in $S$ with a selected literal, then it is true in~$I_S$.\\
(v) If $C$ is non-productive, then it is true in $I_S$.
\end{lem}
\begin{proof}
We use induction on $\succ_g$ and assume that (i)--(v) hold for every $g$-clause $D$ of a clause in $S$ with $C\succ_g D$.

\indent (i) Suppose that $C$ is redundant w.r.t.\;$S$. Then there exist $g$-clauses $C_1', \ldots, C_k'$ of clauses $C_1, \ldots, C_k$ in $S$, such that $\{C_1',\ldots,C_k'\}\models C$ and $C\succ_g C_i'$ for all $1\leq i \leq k$. By the induction hypothesis,  each $C_i'$, $1\leq i \leq k$, is true in $I_S$. Thus, $C$ is true in $I_S$.

\indent (ii) Suppose that $C$ is a $g$-clause of a clause $B:=s_1\approx t_1\vee\cdots\vee s_m\approx t_m \vee u_1\not\approx v_1 \vee \cdots\vee u_n\not\approx v_n$ in $S$ but is not a reduced $g$-clause w.r.t.\;$R_S$. Then $C$ is of the form $C:=s_1w_1\bot\approx t_1w_1\bot\vee\cdots\vee s_mw_m\bot\approx t_mw_m\bot \vee u_1\bot\not\approx v_1\bot \vee \cdots\vee u_n\bot\not\approx v_n\bot$ for $w_1,\ldots,w_m\in \Sigma^*$ and some $w_k\bot$ is reducible by $R_S$. Now, consider $C'=s_1w_1'\bot\approx t_1w_1'\bot\vee\cdots\vee s_mw_m'\bot\approx t_mw_m'\bot \vee u_1\bot\not\approx v_1\bot \vee \cdots\vee u_n\bot\not\approx v_n\bot$, where $w_i'\bot$ is the normal form of $w_i\bot$ w.r.t.\;$R_S$ for each $1\leq i \leq m$. Then $C'$ is a reduced $g$-clause of $B$ w.r.t.\;$R_S$, and is true in $I_S$ by the induction hypothesis. Since each $w_i\bot \approx w_i'\bot$, $1\leq i \leq m$, is true in $I_S$ by Lemma~\ref{lem:confluent}(iii), we may infer that $C$ is true in $I_S$.\\
\indent In the remainder of the proof of this lemma, we assume that $C$ is neither redundant w.r.t.\;$S$ nor is it a reducible $g$-clause w.r.t.\;$R_S$ of some clause in $S$. (Otherwise, we are done by (i) or (ii).)

\indent (iii) Suppose that $C:=C' \vee s\bot \approx t\bot$ produces the rule $s\bot \rightarrow t\bot$.  Since $s\bot \rightarrow t\bot \in E_C \subset R_S$, we see that $C$ is true in $I_S$. We show that $C^\prime$ is false in $I_S$. Let $C':=\Gamma \rightarrow \Delta$. Then $I_C\not\models C'$ by Definition~\ref{defn:model}, which implies that $I_C\cap \Delta = \emptyset, I_C\supseteq \Gamma$, and thus $I_S \supseteq \Gamma$. It remains to show that $I_S\cap \Delta = \emptyset$. Suppose to the contrary that $\Delta$ contains an equation $s'\bot \approx t'\bot$ which is true in $I_S$. Since $I_C\cap\Delta =\emptyset$, we must have $s'\bot\approx t'\bot \in I\setminus I_C$, which is only possible if $s\bot=s'\bot$ and $I_C\models t\bot\approx t'\bot$, contradicting condition (2) in Definition~\ref{defn:model}.

\indent (iv) Suppose that $C$ is of the form $C:=B' \vee s\bot \not\approx t\bot$, where $s\bot \not\approx t\bot$ is a $g$-literal of a selected literal in a clause in $S$ and $B'$ is a $g$-clause of $B$.\\
\indent (iv.1) If $s\bot=t\bot$, then $B'$ is an equality resolvent of $C$ and the Equality Resolution inferences can be lifted. By saturation of $S$ under $\mathfrak{S}$ and the induction hypothesis, $B'$ is true in $I_S$. Thus, $C$ is true in $I_S$.\\
\indent (iv.2) If $s\bot\neq t\bot$, then suppose to the contrary that $C$ is false in $I_S$. Then we have $I_S \models s\bot \approx t\bot$, which implies that $s\bot$ or $t\bot$ is reducible by $R_S$ by Lemma~\ref{lem:confluent}(iii). Without loss of generality, we assume that $s\bot$ is reducible by $R_S$ with some rule $lu\bot \rightarrow ru\bot$ for some $u\in\Sigma^*$ produced by a productive $g$-clause $D' \vee lu\bot \approx ru\bot$ of a clause $D \vee l \approx r \in S$. This means that $s\bot$ has a suffix $lu\bot$. Now, consider the following inference by Rewriting:
\begin{prooftree}
\AxiomC{$B\vee s[lu]_{suf}\not\approx t$}\AxiomC{$D\vee l\approx r$}
\BinaryInfC{$B \vee D\vee s[ru]_{suf}\not\approx t$}
\end{prooftree}
where $s[lu]_{suf}\not\approx t$ is selected for the left premise. The conclusion of the above inference has a $g$-clause $C':=B' \vee D' \vee s\bot[ru\bot]_{suf} \not\approx t\bot$. By saturation of $S$ under $\mathfrak{S}$ and the induction hypothesis, $C'$ must be true in $I_S$. Moreover, we see that $s\bot[ru\bot]_{suf} \not\approx t\bot$ is false in $I_S$ by Lemma~\ref{lem:confluent}(iii), and $D'$ are false in $I_S$ by (iii). This means that $B'$ is true in $I_S$, and thus $C$ (i.e., $C=B' \vee s\bot \not\approx t\bot$) is true in $I_S$, which is the required contradiction.

\indent (v) If $C$ is non-productive, then we assume that $C$ is not a $g$-clause of a clause with a selected literal. Otherwise, the proof is done by (iv). This means that $C$ is of the form $C:=B' \vee su\bot \approx tu\bot$, where $su\bot \approx tu\bot$ is maximal in $C$ and $B'$ contains no selected literal. If $su\bot = tu\bot$, then we are done. Therefore, without loss of generality, we assume that $su\bot \succ_g tu\bot$. As $C$ is non-productive, it means that (at least) one of the conditions in Definition~\ref{defn:model} does not hold.\\
\indent If condition (1) does not hold, then $I_C \models C$, so we have $I_S \models C$, i.e., $C$ is true in $I_S$. If condition (1) holds but condition (2) does not hold, then $C$ is of the form $C:=B_1' \vee su\bot\approx tu\bot \vee svw\bot\approx t'vw\bot$, where $su=svw$ (i.e.,\;$u=vw$) and $I_C \models tu\bot\approx t'vw\bot$.\\
\indent Suppose first that $tu\bot = t'vw\bot$. Then we have $t=t'$ since $u=vw$. Now, consider the following inference by Factoring:
\begin{prooftree}
\AxiomC{$B_1\vee s\approx t \vee sv\approx tv$}
\UnaryInfC{$B_1\vee sv\approx tv$}
\end{prooftree}
The conclusion of the above inference has a $g$-clause $C':=B_1' \vee svw\bot\approx tvw\bot$, i.e.,\;$C':=B_1' \vee su\bot\approx tu\bot$ since $u=vw$. By saturation of $S$ under $\mathfrak{S}$ and the induction hypothesis, $C'$ is true in $I_S$, and thus $C$ is true in $I_S$.\\
\indent Otherwise, suppose that $tu\bot \neq t'vw\bot$. Then we have $tu\bot \downarrow_{R_C} t'vw\bot$ by Lemma~\ref{lem:confluent}(iii) and $tu\bot \succ_g t'vw\bot$ because $su\bot \approx tu\bot$ is maximal in $C$. This means that $tu\bot$ is reducible by $R_C$ by some rule $l\tau\bot \rightarrow r\tau\bot$ produced by a productive $g$-clause $D' \vee l\tau\bot \approx r\tau\bot$ of a clause $D \vee l \approx r \in S$. Now, we need to consider two cases: 

\indent (v.1) If $t$ has the form $t:=u_1u_2$ and $l$ has the form $l:=u_2u_3$, then consider the following inference by Paramodulation:
\begin{prooftree}
\AxiomC{$B\vee s\approx u_1u_2$}\AxiomC {$D\vee u_2u_3\approx r$}
\BinaryInfC{$B\vee D \vee su_3\approx u_1r$}
\end{prooftree}
The conclusion of the above inference has a $g$-clause $C':=B' \vee D' \vee su_3\tau\bot \approx u_1r\tau\bot$ with $u=u_3\tau$. By saturation of $S$ under $\mathfrak{S}$ and the induction hypothesis, $C'$ is true in $I_S$. Since $D'$ is false in $I_S$ by (iii), either $B'$ or $su_3\tau\bot \approx u_1r\tau\bot$ is true in $I_S$. If $B'$ is true in $I_S$, so is $C$. If $su_3\tau\bot \approx u_1r\tau\bot$ is true in $I_S$, then $su\bot \approx tu\bot$ is also true in $I_S$ by Lemma~\ref{lem:confluent}(iii), where $t=u_1u_2$ and $u=u_3\tau$. Thus, $C$ is true in $I_S$.

\indent (v.2) If $t$ has the form $t:=u_1u_2u_3$ and $l$ has the form $l:=u_2$, then consider the following inference by Rewrite:
\begin{prooftree}
\AxiomC{$B\vee s\approx u_1u_2u_3$}\AxiomC {$D\vee u_2\approx r$}
\BinaryInfC{$B\vee D\vee s\approx u_1ru_3$}
\end{prooftree}
The conclusion of the above inference has a $g$-clause $C'':=B' \vee D' \vee su\bot \approx u_1ru_3u\bot$ with $\tau = u_3u$. By saturation of $S$ under $\mathfrak{S}$ and the induction hypothesis, $C''$ is true in $I_S$. Since $D'$ is false in $I_S$ by (iii), either $B'$ or $su\bot \approx u_1ru_3u\bot$ is true in $I_S$. Similarly to case (v.1), if $B'$ is true in $I_S$, so is $C$. If $su\bot \approx u_1ru_3u\bot$ is true in $I_S$, then $su\bot \approx tu\bot$ is also true in $I_S$ by Lemma~\ref{lem:confluent}(iii), where $t=u_1u_2u_3$. Thus, $C$ is true in $I_S$.\\
\indent If conditions (1) and (2) hold but condition (3) does not hold, then $su\bot \approx tu\bot$ is only maximal but is not strictly maximal, so we are in the previous case. (Since $\succ_g$ is total on $g$-clauses, condition (2) does not hold.) If conditions (1)--(3) hold but condition (4) does not hold, then $su\bot$ is reducible by $R_C$ by some rule $l\tau\bot \rightarrow r\tau\bot$ produced by a productive $g$ clause $D' \vee l\tau\bot \approx r\tau\bot$ of a clause $D \vee l \approx r \in S$. Again, we need to consider two cases:

\indent (v.1') If $s$ has the form $s:=u_1u_2$ and $l$ has the form $l:=u_2u_3$, then consider the following inference by Superposition:
\begin{prooftree}
\AxiomC{$B\vee u_1u_2 \approx t$}\AxiomC {$D\vee u_2u_3\approx r$}
\BinaryInfC{$B\vee D \vee u_1r \approx tu_3$}
\end{prooftree}
The conclusion of the above inference has a $g$-clause $C':=B' \vee D' \vee u_1r\tau\bot \approx tu_3\tau\bot$ with $u=u_3\tau$. By saturation of $S$ under $\mathfrak{S}$ and the induction hypothesis, $C'$ is true in $I_S$. Since $D'$ is false in $I_S$ by (iii), either $B'$ or $u_1r\tau\bot \approx tu_3\tau\bot$ is true in $I_S$. If $B'$ is true in $I_S$, so is $C$. If $u_1r\tau\bot \approx tu_3\tau\bot$ is true in $I_S$, then $su\bot \approx tu\bot$ is also true in $I_S$ by Lemma~\ref{lem:confluent}(iii), where $s=u_1u_2$ and $u=u_3\tau$. Thus, $C$ is true in $I_S$.

\indent (v.2') If $s$ has the form $s:=u_1u_2u_3$ and $l$ has the form $l:=u_2$, then consider the following inference by Rewrite:
\begin{prooftree}
\AxiomC{$B\vee u_1u_2u_3 \approx t$}\AxiomC {$D\vee u_2\approx r$}
\BinaryInfC{$B\vee D\vee u_1ru_3\approx t$}
\end{prooftree}
The conclusion of the above inference has a $g$-clause $C'':=B' \vee D' \vee u_1ru_3u\bot \approx tu\bot$ with $\tau = u_3u$. By saturation of $S$ under $\mathfrak{S}$ and the induction hypothesis, $C''$ is true in $I_S$. Since $D'$ is false in $I_S$ by (iii), either $B'$ or $u_1ru_3u\bot \approx tu\bot$ is true in $I_S$. Similarly to case (v.1'),  If $B'$ is true in $I_S$, so is $C$. If $u_1ru_3u\bot \approx tu\bot$ is true in $I_S$, then $su\bot \approx tu\bot$ is also true in $I_S$ by Lemma~\ref{lem:confluent}(iii), where $s=u_1u_2u_3$. Thus, $C$ is true in $I_S$.
\end{proof}

\begin{defn}\normalfont\label{defn:theoremproving} (i) A \emph{theorem proving derivation} is a sequence of sets of clauses $S_0=S,S_1,\ldots$ over $\Sigma^*$ such that:\\
\indent (i.1) Deduction: $S_i=S_{i-1} \cup \{C\}$ if $C$ can be deduced from premises in $S_{i-1}$ by applying an inference rule in $\mathfrak{S}$.\\
\indent (i.2) Deletion: $S_i = S_{i-1} \setminus \{D\}$ if $D$ is redundant w.r.t.\;$S_{i-1}$.\footnote{Here, an inference by Simplification combines the Deduction step for $C\vee l_1rl_2\bowtie v$ and the Deletion step for $C\vee l_1ll_2\bowtie v$ (see the Simplification rule).}

 (ii) The set $S_\infty:=\bigcup_i(\bigcap_{j\geq i}S_j)$ is the \emph{limit} of the theorem proving derivation.
\end{defn}

\AxiomC{$S\cup\{C\vee l_1ll_2\bowtie v,\; l\approx r\}$}
\UnaryInfC{$S\cup\{C\vee l_1rl_2\bowtie v,\; l\approx r\}$}

We see that the soundness of a theorem proving derivation w.r.t.\;the proposed inference system is straightforward, i.e., $S_i\models S_{i+1}$ for all $i\geq 0$.

\begin{defn}\normalfont A theorem proving derivation $S_0, S_1,S_2,\ldots$ is \emph{fair} w.r.t.\;the inference system $\mathfrak{S}$ if every inference by $\mathfrak{S}$ with premises in $S_\infty$ is redundant w.r.t.\;$\bigcup_j S_j$.
\end{defn}

\begin{lem}\label{lem:inclusion} Let $S$ and $S'$ be sets of clauses over $\Sigma^*$.\\
(i) If $S\subseteq S'$, then any clause which is redundant w.r.t.\;$S$ is also redundant w.r.t.\;$S'$.\\
(ii) If $S\subseteq S'$ and all clauses in $S'\setminus S$ are redundant w.r.t.\;$S'$, then any clause or inference which is redundant w.r.t.\;$S'$ is also redundant w.r.t.\;$S$.
\end{lem}
\begin{proof} The proof of part (i) is obvious. For part (ii), suppose that a clause $C$ is redundant w.r.t.\;$S'$ and let $C'$ be a $g$-clause of it. Then there exists a minimal set $N:=\{C_1', \ldots, C_n'\}$ (w.r.t.\;$\succ_g$) of $g$-clauses of clauses in $S'$ such that $N\models C'$ and $C'\succ_g C_i'$ for all $1\leq i \leq n$. We claim that all $C_i'$ in $N$ are not redundant w.r.t.\;$S'$, which shows that $C'$ is redundant w.r.t.\;$S$. Suppose to the contrary that some $C_j'$ is redundant w.r.t.\;$S'$. Then there exist a set $N':=\{D_1', \ldots, D_m'\}$ of $g$-clauses of clauses in $S'$ such that $N'\models C_j'$ and $C_j'\succ_g D_i'$ for all $1\leq i \leq m$. This means that we have $\{C_1',\ldots, C_{j-1}', D_1',\ldots, D_m', C_{j+1}',\ldots, C_n'\}\models C'$, which contradicts our minimal choice of the set $N=\{C_1', \ldots, C_n'\}$.\\
\indent Next, suppose an inference $\pi$ with conclusion $D$ is redundant w.r.t.\;$S'$ and let $\pi'$ be a $g$-instance of it such that $B$ is the maximal premise and $D'$ is the conclusion of $\pi'$ (i.e.,\;a $g$-clause of $D$). Then there exists a minimal set $P:=\{D_1', \ldots, D_n'\}$ (w.r.t.\;$\succ_g$) of $g$-clauses of clauses in $S'$ such that $P\models D'$ and $B\succ_g D_i'$ for all $1\leq i \leq n$. As above, we may infer that all $D_i'$ in $P$ are not redundant w.r.t.\;$S'$, and thus $\pi'$ is redundant w.r.t.\;$S$.
\end{proof}

\begin{lem}\label{lem:fairderivation} Let $S_0,S_1,\ldots$ be a fair theorem proving derivation w.r.t.\;$\mathfrak{S}$. Then $S_\infty$ is saturated under $\mathfrak{S}$.
\end{lem}
\begin{proof}
If $S_\infty$ contains the empty clause, then it is obvious that $S_\infty$ is saturated under $\mathfrak{S}$. Therefore, we assume that the empty clause is not in $S_\infty$.\\
\indent If a clause $C$ is deleted in a theorem proving derivation, then $C$ is redundant w.r.t.\;some $S_j$. By Lemma~\ref{lem:inclusion}(i), it is also redundant w.r.t.\;$\bigcup_jS_j$. Similarly, every clause in $\bigcup_jS_j\setminus S_\infty$ is redundant w.r.t.\;$\bigcup_jS_j$.\\
\indent By fairness, every inference $\pi$ by $\mathfrak{S}$ with premises in $S_\infty$ is redundant w.r.t.\;$\bigcup_j S_j$. Using Lemma~\ref{lem:inclusion}(ii) and the above, $\pi$ is also redundant w.r.t.\;$S_\infty$, which means that $S_\infty$ is saturated under $\mathfrak{S}$.
\end{proof}

\begin{thm}\label{thm:sat} Let $S_0,S_1,\ldots$ be a fair theorem proving derivation w.r.t.\;$\mathfrak{S}$. If $S_\infty$ does not contain the empty clause, then $I_{S_\infty} \models S_0$ (i.e., $S_0$ is satisfiable.)
\end{thm}
\begin{proof}
Suppose that $S_0,S_1,\ldots$ is a fair theorem proving derivation w.r.t.\;$\mathfrak{S}$ and that its limit $S_\infty$ does not contain the empty clause. Then $S_\infty$ is saturated under $\mathfrak{S}$ by Lemma~\ref{lem:fairderivation}. Let $C'$ be a $g$-clause of a clause $C$ in $S_0$. If $C\in S_\infty$, then $C'$ is true in $I_{S_\infty}$ by Lemma~\ref{lem:model}. Otherwise, if $C\notin S_\infty$, then $C$ is redundant w.r.t.\;some $S_j$. It follows that $C$ redundant w.r.t.\;$\bigcup_j S_j$ by Lemma~\ref{lem:inclusion}(i), and thus redundant w.r.t.\;$S_\infty$ by Lemma~\ref{lem:inclusion}(ii). This means that there exist $g$-clauses $C_1', \ldots, C_k'$ of clauses $C_1, \ldots, C_k$ in $S_\infty$ such that $\{C_1',\ldots,C_k'\}\models C'$ and $C'\succ_g C_i'$ for all $1\leq i \leq k$. Since each $C_i'$, $1\leq i\leq k$, is true in $I_{S_\infty}$ by Lemma~\ref{lem:model}, $C'$ is also true in $I_{S_\infty}$, and thus the conclusion follows.
\end{proof}

The following theorem states that $\mathfrak{S}$ with the contraction rules is refutationally complete for clauses over $\Sigma^*$.

\begin{thm}\label{thm:maintheorem} Let $S_0,S_1,\ldots$ be a fair theorem proving derivation w.r.t.\;$\mathfrak{S}$. Then $S_0$ is unsatisfiable if and only if the empty clause is in some $S_j$.
\end{thm}
\begin{proof}
Suppose that $S_0,S_1,\ldots$ be a fair theorem proving derivation w.r.t.\;$\mathfrak{S}$. By the soundness of the derivation, if the empty clause is in some $S_j$, then $S_0$ is unsatisfiable. Otherwise, if the empty clause is not in $S_k$ for all $k$, then $S_\infty$ does not contain the empty clause by the soundness of the derivation. Applying Theorem~\ref{thm:sat}, we conclude that $S_0$ is satisfiable.
\end{proof}

\section{Conditional Completion}\label{sec:HornCompletion}
In this section, we present a saturation procedure under $\mathfrak{S}$ for a set of conditional equations over $\Sigma^*$, where a conditional equation is naturally written as an equational Horn clause. A saturation procedure under $\mathfrak{S}$ can be viewed as \emph{conditional completion}~\cite{Dershowitz1991} for a set of conditional equations over $\Sigma^*$. If a set of conditional equations over $\Sigma^*$ is simply a set of equations over $\Sigma^*$, then the proposed saturation procedure (w.r.t.\;$\succ$) corresponds to a completion procedure for a string rewriting system. Conditional string rewriting systems were considered in~\cite{Deiss1992} in the context of embedding a finitely generated monoid with decidable word problem into a monoid presented by a finite convergent conditional presentation. It neither discusses a conditional completion (or a saturation) procedure, nor considers the word problems for conditional equations over $\Sigma^*$ in general.\\
\indent First, it is easy to see that a set of equations over $\Sigma^*$ is consistent. Similarly, a set of conditional equations $R$ over $\Sigma^*$ is consistent because each conditional equation has always a positive literal and we cannot derive the empty clause from $R$ using a saturation procedure under $\mathfrak{S}$ that is refutationally complete (cf.~Section 9 in~\cite{Dershowitz2001}). Since we only consider Horn clauses in this section, we neither need to consider the Factoring rule nor the Paramodulation rule in $\mathfrak{S}$. In the remainder of this section, by a conditional equational theory $R$, we mean a set of conditional equations $R$ over $\Sigma^*$. 

\begin{defn}\normalfont Given a conditional equational theory $R$ and two finite words $s,t \in \Sigma^*$, a \emph{word problem} w.r.t.\;$R$ is of the form $\phi:= s\approx^?t$. The \emph{goal} of this word problem is $s\not\approx t$. We say that a word problem $s\approx^?t$ w.r.t.\;$R$ is \emph{decidable} if there is a decision procedure for determining whether $s\approx t$ is entailed by $R$ (i.e., $R\models s\approx t$) or not (i.e., $R\not\models s\approx t$) .
\end{defn} 

Given a conditional equational theory $R$, let $G:=s\not\approx t$ be the goal of a word problem $s\approx^? t$ w.r.t.\;$R$. (Note that $G$ does not have any positive literal.) Then we see that $R \cup \{s\approx t\}$ is consistent if and only if $R\cup \{G\}$ is inconsistent.  This allows one to decide a word problem w.r.t.\;$R$ using the equational theorem proving procedure discussed in Section~\ref{sec:completeness}. 

\begin{lem}\label{lem:wordproblem} Let $R$ be a conditional equational theory finitely saturated under $\mathfrak{S}$. Then Rewrite together with Equality Resolution is terminating and refutationally complete for $R\cup \{G\}$, where $G$ is the goal of a word problem w.r.t.\;$R$.
\end{lem}
\begin{proof} Since $R$ is already saturated under $\mathfrak{S}$, inferences among Horn clauses in $R$ are redundant and remain redundant in $R\cup\{G\}$ for a theorem proving derivation starting with $R\cup\{G\}$. (Here, $\{G\}$ can be viewed as a \emph{set of support}~\cite{Bachmair1994} for a refutation of $R\cup\{G\}$.) Now, observe that $G$ is a negative literal, so it should be selected. The only inference rules in $\mathfrak{S}$ involving a selected literal are the Rewrite and Equality Resolution rule. Furthermore, the derived literals from $G$ w.r.t.\;Rewrite will also be selected eventually. Therefore, it suffices to consider positive literals as the right premise (because they contain no selected literal), and $G$ and its derived literals w.r.t.\;Rewrite as the left premise for the Rewrite rule. Observe also that if $G'$ is an immediate derived literal from $G$ w.r.t.\;Rewrite, then we see that $G\succ G'$. If $G$ or its derived literal from $G$ w.r.t.\;Rewrite becomes of the form $u\not\approx u$ for some $u\in \Sigma^*$, then it will also be selected and an Equality Resolution inference yields the empty clause. Since $\succ$ is terminating and there are only finitely many positive literals in $R$, we may infer that the Rewrite and Equality Resolution inference steps on $G$ and its derived literals are terminating. (The number of positive literals in $R$ remains the same during a theorem proving derivation starting with $R\cup\{G\}$ using our selection strategy.)\\
\indent Finally, since $\mathfrak{S}$ is refutationally complete by Thereom~\ref{thm:maintheorem}, Rewrite together with Equality Resolution is also refutationally complete for $R\cup\{G\}$.
\end{proof}

Given a finitely saturated conditional equational theory $R$ under $\mathfrak{S}$, we provide a decision procedure for the word problems w.r.t.\;$R$ in the following theorem.

\begin{thm}\label{thm:main2} Let $R$ be a conditional equational theory finitely saturated under $\mathfrak{S}$. Then the word problems w.r.t.\;$R$ are decidable by Rewrite together with Equality Resolution.
\end{thm}
\begin{proof}
Let $\phi:=s\approx^?t$ be a word problem w.r.t.\;$R$ and $G$ be the goal of $\phi$. We know that by Lemma~\ref{lem:wordproblem}, Rewrite together with Equality Resolution is terminating and refutationally complete for $R\cup \{G\}$. Let $R_0:=R\cup\{G\}, R_1, \ldots, R_n$ be a fair theorem proving derivation w.r.t.\;Rewrite together with Equality Resolution such that $R_n$ is the limit of this derivation. If $R_n$ contains the empty clause, then $R_n$ is inconsistent, and thus $R_0$ is inconsistent, i.e.,\;$\{s\not\approx t\}\cup R$ is inconsistent by the soundness of the derivation. Since $R$ is consistent and $\{s\not\approx t\}\cup R$ is saturated under $\mathfrak{S}$, we may infer that $R\models s\approx t$. 

Otherwise, if $R_n$ does not contain the empty clause, then $R_n$ is consistent, and thus $R_0$ is consistent by Theorem~\ref{thm:maintheorem}, i.e.,\;$\{s\not\approx t\} \cup R$ is consistent. Since $R$ is consistent and $\{s\not\approx t\} \cup R$ is saturated under $\mathfrak{S}$, we may infer that $R\not\models s\approx t$.
\end{proof}

The following corollary is a consequence of Theorem~\ref{thm:main2} and the following observation. Let $R=R_0, R_1, \ldots, R_n$ be a finite fair theorem proving derivation w.r.t.\;$\mathfrak{S}$ for an initial conditional equational theory $R$ with the limit $\bar{R}:=R_n$. Then $R\cup \{G\}$ is inconsistent if and only if $\bar{R}\cup \{G\}$ is inconsistent by the soundness of the derivation and Theorem~\ref{thm:maintheorem}.

\begin{cor} Let $R=R_0, R_1, \ldots$ be a fair theorem proving derivation w.r.t.\;$\mathfrak{S}$ for a conditional equational theory $R$. If $R$ can be finitely saturated under $\mathfrak{S}$, then the word problems w.r.t.\;$R$ are decidable.
\end{cor}

\begin{exam}\normalfont
\label{ex:ex3}
Let $a \succ b \succ c$ and $R$ be a conditional equational theory consisting of the following rules 1: $aa \approx \lambda$, 2: $bb \approx \lambda$, 3: $ab \approx \lambda$, 4: $ab\not\approx ba \vee ac \approx ca$, and 5: $ab\not\approx ba \vee ac \not\approx ca \vee bc \approx cb$. We first saturate $R$ under $\mathfrak{S}$:\\\\
6: $\lambda \not\approx ba \vee ac\approx ca$ ($ab\not\approx  ba$ is selected for 4. Rewrite of 4 with 3)\\
7:  $\lambda\not\approx ba \vee ac \not\approx ca \vee bc \approx cb$ ($ab\not\approx ba$ is selected for 5. Rewrite of 5 with 3)\\
8: $a\approx b$ (Superposition of 1 with 3)\\
9: $\lambda \not\approx bb \vee ac\approx ca$ ($\lambda \not\approx ba$ is selected for 6. Rewrite of 6 with 8)\\
10: $\lambda \not\approx \lambda \vee ac\approx ca$ ($\lambda \not\approx bb$ is selected for 9. Rewrite of 9 with 2)\\
11: $ac\approx ca$ ($\lambda \not\approx \lambda$ is selected for 10. Equality Resolution on 10)\\
12: $\lambda\not\approx bb \vee ac \not\approx ca \vee bc \approx cb$ ($\lambda\not\approx ba$ is selected for 7. Rewrite of 7 with 8)\\
13: $\lambda\not\approx \lambda \vee ac \not\approx ca \vee bc \approx cb$ ($\lambda\not\approx bb$ is selected for 12. Rewrite of 12 with 2)\\
14: $ac \not\approx ca \vee bc \approx cb$ ( $\lambda\not\approx \lambda$ is selected for 13. Equality Resolution on 13)\\
15: $ca \not\approx ca \vee bc \approx cb$ ( $ac \not\approx ca$ is selected for 14. Rewrite of 14 with 11)\\
16: $bc \approx cb$ ($ca \not\approx ca$ is selected for 15. Equality Resolution on 15)\\
$\cdots$\\
After some simplification steps, we have a saturated set $\bar{R}$ for $R$ under $\mathfrak{S}$ using our selection strategy (i.e., the selection of negative literals). We may infer that the positive literals in $\bar{R}$ are as follows. $1^\prime:bb \approx \lambda$, $2^\prime:a\approx b$, and $3^\prime: bc \approx cb$. Note that only the positive literals in $\bar{R}$ are now needed to solve a word problem w.r.t.\;$R$ because of our selection strategy.\\
\indent  Now, consider the word problem $\phi:= acbcba\approx^?bccaba$ w.r.t.\;$R$, where the goal of $\phi$ is $G:=acbcba\not\approx bccaba$. We only need the Rewrite and Equality Resolution steps on $G$ and its derived literals from $G$ using $1^\prime$, $2^\prime$, and $3^\prime$. Note that all the following literals are selected except the empty clause.\\\\
$4^\prime$: $bcbcbb\not\approx bccbbb$ (Rewrite steps of $G$ and its derived literals from $G$ using $2^\prime$).\\
$5^\prime$: $bcbc\not\approx bccb$ (Rewrite steps of $4^\prime$ and its derived literals from $4^\prime$  using $1^\prime$).\\
$6^\prime$: $ccbb\not\approx ccbb$ (Rewrite steps of $5^\prime$ and its derived literals from $5^\prime$ using $3^\prime$).\\
$7^\prime$: $\square$ (Equality Resolution on $6^\prime$)\\\\
\indent Since $\bar{R}\cup G$ is inconsistent, we see that $R\cup G$ is inconsistent by the soundness of the derivation, where $R$ and $\bar{R}$ are consistent. Therefore, we may infer that $R \models acbcba\approx bccaba$.
\end{exam}

\section{Related Work}
Equational reasoning on strings has been studied extensively in the context of string rewriting systems and Thue systems~\cite{Book1993} and their related algebraic structures. The monotonicity assumption used in this paper is found in string rewriting systems and Thue systems in the form of a congruence relation (see~\cite{Book1993,Kapur1985}). See~\cite{Book1981,Madlener1991,Cremanns2002,Otto1997} also for the completion of algebraic structures and decidability results using string rewriting systems, in particular \emph{cross-sections} for finitely presented monoids discussed by Otto et al~\cite{Otto1997}. However, those systems are not concerned with equational theorem proving for general clauses over strings. If the monotonicity assumption is discarded, then equational theorem proving for clauses over strings can be handled by traditional superposition calculi or SMT with the theory of equality with uninterpreted functions (EUF) and their variants~\cite{Barrett2009} using a simple translation into first-order ground terms. Also, efficient SMT solvers for various string constraints were discussed in the literature (e.g.,~\cite{Liang2016}).

Meanwhile, equational theorem proving modulo associativity was studied in~\cite{Rubio1996}. (See also~\cite{Kutsia2002} for equational theorem proving with \emph{sequence variables} and fixed or variadic arity symbols). This approach is not tailored towards (ground) strings, so we need an additional encoding for each string. Also, it is probably less efficient since it is not tailored towards ground strings, and it does not provide a similar decision procedure discussed in Section~\ref{sec:HornCompletion}.

The proposed calculus is the first sound and refutationally complete equational theorem proving calculus for general clauses over strings under the monotonicity assumption. One may attempt to use the existing superposition calculi for clauses over strings with the proposed translation scheme, which translates clauses over strings into clauses over first-order terms discussed in Section~\ref{subsec:lifting}. However, this does not work because of the Equality Factoring rule~\cite{Bachmair1994,Nieuwenhuis2001} or the Merging Paramodulation rule~\cite{Bachmair1994}, which is essential for first-order superposition theorem proving calculi in general. For example, consider a clause $a\approx b \vee a\approx c$ with $a\succ b\succ c$, which is translated into a first-order clause $a(x)\approx b(x) \vee a(y)\approx c(y)$. The Equality Factoring rule yields $b(z)\not\approx c(z) \vee a(z)\approx c(z)$ from  $a(x)\approx b(x) \vee a(y)\approx c(y)$, which cannot be translated back into a clause over strings (see Lemma~\ref{lem:representation}). Similarly, a first-order clause produced by Merging Paramodulation may not be translated back into a clause over strings. If one is only concerned with refutational completeness, then the existing superposition calculi\footnote{The reader is also encouraged to see \emph{AVATAR modulo theories}~\cite{Reger2016}, which is based on the concept of splitting.} can be adapted by using the proposed translation scheme. In this case, a saturated set may not be translated back into clauses over strings in some cases, which is an obvious drawback for its applications (see \emph{programs} in~\cite{Bachmair1994}).

\section{Conclusion}
This paper has presented a new refutationally complete superposition calculus with strings and provided a framework for equational theorem proving for clauses over strings. The results presented in this paper generalize the results about completion of string rewriting systems and equational theorem proving using equations over strings. The proposed superposition calculus is based on the simple string matching methods and the efficient length-lexicographic ordering that allows one to compare two finite strings in linear time for a fixed signature with its precedence.\\
\indent The proposed approach translates for a clause over strings into the first-order representation of the clause by taking the monotonicity property of equations over strings into account. Then the existing notion of redundancy and model construction techniques for the equational theorem proving framework for clauses over strings has been adapted. This paper has also provided a decision procedure for word problems over strings w.r.t.\;a set of conditional equations $R$ over strings if $R$ can be finitely saturated under the Superposition, Rewrite and Equality Resolution rule. (The complexity analysis of the proposed approach is not discussed in this paper. It is left as a future work for this decision procedure.)

\indent Since strings are fundamental objects in mathematics, logic, and computer science including formal language theory, developing applications based on the proposed superposition calculus with strings may be a promising future research direction. Also, the results in this paper may have potential applications in verification systems and solving satisfiability problems~\cite{Armando2003}.\\
\indent In addition, it would be an interesting future research direction to extend our superposition calculus with strings to superposition calculi with strings using built-in equational theories, such as commutativity, \emph{idempotency}~\cite{Book1993}, \emph{nilpotency}~\cite{Guo1996},  and their various combinations. For example, research on superposition theorem proving for \emph{commutative monoids}~\cite{Rosales1999} is one such direction.

\bibliographystyle{eptcs}
\bibliography{references}
\end{document}